\documentclass[12pt, draftclsnofoot, onecolumn]{IEEEtran}

\usepackage[utf8]{inputenc}
\usepackage{amssymb}

\usepackage{enumitem}

\usepackage[most]{tcolorbox}
\tcbset{on line, 
        boxsep=4pt, left=0pt,right=0pt,top=0pt,bottom=0pt,
        colframe=white,colback=gray!10,  
        highlight math style={enhanced}
        }

\usepackage{amsbsy}
\usepackage{bm}
\usepackage{fixmath}
\usepackage{silence}
\usepackage{pst-node,pst-plot}
\usetikzlibrary{shapes.geometric}
\usepackage{mathdots}
\usetikzlibrary{math}

\usepackage{upgreek}

\usepackage{autobreak}
\allowdisplaybreaks

\usepackage[caption=false]{subfig}
\usepackage{commath,amsmath,amssymb,amsfonts}
\usepackage{mathtools}
\usepackage{amsthm}
\usepackage{pgfplots} 
\pgfplotsset{compat=1.15}
\usepackage{graphicx}
\usepackage{pgfgantt}
\usepackage{pdflscape}
\usepackage{pst-plot}
\usepackage{xfrac}
\usepackage{colortbl}
\usepackage{cancel}
\usepgfplotslibrary{fillbetween}
\usepackage{amssymb,bm}
\usepackage{float}
\usepackage{amsmath}
\usepackage[draft=false]{hyperref}
\usepackage{multirow}
\usepackage{xcolor}
\usepackage{mathrsfs}
\usepackage{bbm}

\newcommand{\ceil}[1]{\left \lceil #1 \right \rceil}

\usepackage{tikz}
\usetikzlibrary{arrows.meta}
\usetikzlibrary{positioning}
\usetikzlibrary{shapes.geometric, arrows}
\usetikzlibrary{shapes.geometric}

\usepackage{tikz}
\usetikzlibrary{positioning}
\usetikzlibrary{shapes.geometric, arrows}

\tikzstyle{startstop} = [rectangle, dashed, minimum width=1cm, minimum height=1cm,text centered, draw=black]
\tikzstyle{process} = [rectangle, minimum width=1.5cm, minimum height=1cm, text centered, draw=black]
\tikzstyle{colrec} = [rectangle, minimum width=1cm, minimum height=1cm, text centered, draw=black]
\tikzstyle{rec} = [rectangle, minimum width=1cm, minimum height=1cm, text centered, draw=black, text=purple]
\tikzstyle{arrow} = [thick,->,>=stealth]
\tikzstyle{arrows} = [thick,->,>=stealth]

\usepackage{cite} 

\usepackage{comment}

\usepackage{autobreak}
\allowdisplaybreaks

\def \C{\mathcal{C}}
\def \D{\mathcal{D}}

\def \U{\mathcal{U}}

\def \X{\mathcal{X}}

\def \fx{\mathbf{x}}
\def \fy{\mathbf{y}}

\def \fY{\mathbf{Y}}

\def \f0{\mathbf{0}}

\definecolor{blau_1a}{RGB}{93,133,195}
\definecolor{blau_2a}{RGB}{0,156,218}
\definecolor{gruen_3a}{RGB}{80,182,149}
\definecolor{gruen_4a}{RGB}{175,204,80}
\definecolor{gruen_5a}{RGB}{221,223,72}
\definecolor{orange_6a}{RGB}{255,224,92}
\definecolor{orange_7a}{RGB}{248,186,60}
\definecolor{rot_8a}{RGB}{238,122,52}
\definecolor{rot_9a}{RGB}{233,80,62}
\definecolor{lila_10a}{RGB}{201,48,142}
\definecolor{lila_11a}{RGB}{128,69,151}

\definecolor{blau_1b}{RGB}{0,90,169}
\definecolor{blau_2b}{RGB}{0,131,204}
\definecolor{gruen_3b}{RGB}{0,157,129}
\definecolor{gruen_4b}{RGB}{153,192,0}
\definecolor{gruen_5b}{RGB}{201,212,0}
\definecolor{orange_6b}{RGB}{253,202,0}
\definecolor{orange_7b}{RGB}{245,163,0}
\definecolor{rot_8b}{RGB}{236,101,0}
\definecolor{rot_9b}{RGB}{230,0,26}
\definecolor{lila_10b}{RGB}{166,0,132}
\definecolor{lila_11b}{RGB}{114,16,133}

\definecolor{mycolor1}{rgb}{0.0, 0.18, 0.39}
\definecolor{mycolor2}{RGB}{87,108,67}
\definecolor{mycolor3}{RGB}{8,133,161}
\definecolor{mycolor4}{RGB}{80,91,161}
\definecolor{mycolor5}{RGB}{98,122,157}
\definecolor{mycolor6}{RGB}{255,163,67}
\definecolor{mycolor7}{RGB}{152,205,225}
\definecolor{mycolor8}{RGB}{242,204,48}
\definecolor{mycolor9}{rgb}{0,.5,0}
\definecolor{mycolor10}{rgb}{.59,.44,.09}
\definecolor{mycolor11}{RGB}{231,199,31} 
\definecolor{mycolor12}{RGB}{8,133,161} 
\definecolor{mycolor13}{RGB}{157,188,64} 
\definecolor{mycolor14}{RGB}{194,150,130} 
\definecolor{mycolor15}{RGB}{98,122,157} 
\definecolor{mycolor16}{RGB}{160,160,160} 
\definecolor{mycolor17}{RGB}{115,82,68} 
\definecolor{mycolor18}{RGB}{94,60,108} 
\definecolor{mycolor19}{RGB}{115,82,68} 
\definecolor{mycolor20}{RGB}{255,183,30} 

\usepackage{accents}

\theoremstyle{remark} \newtheorem{theorem}{Theorem}
\theoremstyle{remark} 
\theoremstyle{remark} \newtheorem{corollary}[theorem]{Corollary}
\theoremstyle{remark} 
\theoremstyle{remark} \newtheorem{definition}{Definition}
\theoremstyle{remark} \newtheorem{remark}{Remark}
\theoremstyle{remark}

\usepackage{amssymb}

\usepackage{pst-node,pst-plot}

\usetikzlibrary{shapes.geometric}

\usepackage{tikz}
\usetikzlibrary{arrows.meta}
\usetikzlibrary{arrows, arrows.meta}
\usepackage{pgfplots}
\pgfplotsset{compat=1.15}
\usepgfplotslibrary{polar}

\usepackage[bottom]{footmisc}

\usepackage{cleveref}
\crefname{equation}{Eq}{}
\crefrangelabelformat{equation}{(#3#1#4--#5#2#6)}

\usepackage{algpseudocode}
\usepackage[linesnumbered,ruled,noline]{algorithm2e}
\SetKwInput{KwInput}{Input}
\SetKwInput{KwOutput}{Output}
\SetKwBlock{Loop}{Loop}{end}

\usepackage{diffcoeff}
\usepackage{fullpage,amsmath,tkz-base}  
\usepackage{cases}

\usepackage{booktabs}

\usepackage{wrapfig}
\usepackage{upgreek}

\usepackage{siunitx}
\usepackage{framed}

\definecolor{cell}{RGB}{220,230,240}
\definecolor{line}{RGB}{80,130,190}
\usepackage{upgreek}

\usepackage{siunitx}
\usepackage{framed}
\normalsize
\ifCLASSINFOpdf
\else
\fi

\begin{document}
\title{Identification via Functions}
\author{
	\vspace{0.3cm}
    \fontsize{11}{11} \selectfont \IEEEauthorblockN{Mohammad Javad Salariseddigh\IEEEauthorrefmark{1} and Feriel Fendri\IEEEauthorrefmark{1}
    }
    \\
	\vspace{0.45cm}
    \fontsize{11}{11} \selectfont \IEEEauthorblockA{\IEEEauthorrefmark{1} Institute for Communications Engineering, Technical University of Munich
    }
    \\[.8em]
    Emails: \{mjss, feriel.fendri@tum.de\}
\vspace{-4mm}
}

\maketitle

\begin{abstract}
We develop a framework for studying the problem of identifying roots of a noisy function. We revisit a previous logarithmic bound on the number of observations and propose a general problem for identification of roots with three errors. As a key finding, we establish a novel logarithmic lower bound on the number of observations which outperforms the previous result across certain regimes of error and accuracy of the identification test. Furthermore, we recover the previous results for root identification as a special case and draw a connection to the message identification problem of Ahlswede.
\end{abstract}

\section{Introduction}
In the identification problem one is interested to verify reliably whether or not a desired message (e.g., event, graph node, object, etc) has been sent by the transmitter. The identification problem for channels was proposed by Ahlswede in \cite{AD89}. Studying the identification problem with deterministic encoding is explored extensively for basic information theoretical models in \cite{Salariseddigh_PhD_Diss}. Moreover, discussions on its relevance to goal oriented semantic communications can be found in \cite{Salariseddigh23_BSC_Future_Internet}. With regard to the framework in \cite{AD89} where receiver aims to identify a message which has experienced a noisy channel, in identification via functions, one may interpret a message as root of a function whose noisy observations are obtained after measurement. Then, given an arbitrary interval, the task is to determine whether or not the root of the function belongs to such interval. In this work we develop techniques to adopt the identification problem for noisy functions and formulate performance bounds on the number of observations which ensure a reliable identification of the roots.

Identification of roots of a function may find explicit applications in a number of different mathematical and engineering fields, including quantile identification, extrema identification, and finding solution to a general non-linear algebraic equation. Moreover, it can be used in eigenvalue problem with the context of numerical analysis or linear algebra \cite{Hoffmann71} where one may seek to realize whether a certain matrix has eigenvalue in a specific domain or not, before running an algorithm for the computation. In this paper, we develop a generalized test model for identification of roots of a noisy function. In the previous test with two type of errors \cite{KleinewaechterISIT97} when root of a function is coordinated in the silent gap then the test cannot identify the root even by increasing the number of observations arbitrary large, a rather obvious reason for such ill-posed situation is that the root is invisible to the test. This implies that if we devise a new test with different visible ranges, the root would probably be treated as it is in the visible range and one might identify it. However, this might comes in the expense of adding extra parameters related to the extra partitions and accuracy levels corresponding to the new silent gaps. In this work, we examine this approach by adding the number of silent gaps and derive fundamental bounds on the required number of observations of a function for identifying its root.
\subsection{Organization}
This work is structured as follows: In Section~\ref{Sec.II}, we review some of the well-known methods for finding roots of a general function. Section~\ref{Sec.III} provides problem formulation, required preliminaries regarding test functions, and previous results. Section~\ref{Sec.IV} includes main contributions and results. Next, in Section~\ref{Sec.V}, we compare the identification of messages and identification of roots in more detail. Finally, Section~\ref{Sec.VI} concludes with a summary and discussions.

\subsection{Notations}
We use the following notations throughout this paper: Blackboard bold letters $\mathbbmss{K,X,Y,Z}\ldots$ are used for alphabet sets. Lower case letters $x,y,z,\ldots$ stand for constants and values (realization) of random variables (RVs), and upper case letters $X,Y,Z,\ldots$ stand for random variables. Lower case bold symbol $\fx$ and $\fy$ stand for row vectors of size $n,$ that is, $\fx = (x_1, \dots, x_n)$ and $\fy = (y_1, \dots, y_n).$ The distribution of a random variable $X$ is specified by a probability mass function (PMF) $p_X(x)$ over a finite set $\X.$ The set of real line is denoted by $\mathbb{R}.$ The notation $\circ$ is used for composition of functions.
\section{Root Finding Algorithms}
\label{Sec.II}
A classic and major challenge in science and engineering is finding the roots of a nonlinear equation $f(x)=0,$ where $f(x): X \rightarrow X,$ for $X  \subset \mathbb{R}.$ The solutions must then be acquired by using numerical models based on iterative procedures \cite{McNamee13}. In this section, we review some of the conventional methods based on the iterative procedure which address such a problem.

\subsection{Newton-Raphson Method}
\label{Subsec.NR-Method}
The Newton-Raphson (NR) algorithm \cite{Newton} is one of the simplest and widely used method in numerical analysis literature for computation of the root of a non-linear algebraic function $f(x).$ We assume that $f(x)$ is monotonic, differentiable and whose first two derivatives are continuous in an interval $[a,b]$ which contains the root $\upkappa.$ This method begins with a point known as the starting point $x_1.$ To select the next point we approximate the function $f(x)$ with a tangent line passing through $f(x_1)$ which has the same slope as the derivative of the function at $x_1.$ The tangent line is expressed by the equation:
\begin{equation}
    L = (x-x_1)f'(x_1)+f(x_1) ,
\end{equation}
where $f'(\cdot)$ is the derivative of $f(\cdot).$ To find the next suggestion $x_2,$ we must find the point where this tangent line crosses 0, i.e., $L=0.$ We determine this approximation by the formula:
\begin{equation}
    x_{2} = x_1 - \frac{f(x_1)}{f'(x_1)} .
\end{equation}
We perform this procedure multiple times until we converge to a specific root, $\upkappa$ using the following iterative formula:

\begin{equation}
    \label{Eq.NRConv}
    x_{n+1} = x_n - \frac{f(x_n)}{f'(x_n)} .
\end{equation}

The following conditions have to be fulfilled to ensure this convergence:
\begin{enumerate}
    \item $x_{n+1}$ has to be between $x_n$ and $\upkappa.$
    \item There must be no other roots of the function f between $\upkappa$ and $x_1.$
\end{enumerate}

The sequence $x_n$ converge to a value $x^* \in [\upkappa, x_1]$ and, from \eqref{Eq.NRConv} we can assume that $$\lim_{n \to \infty} \frac{f(x_n)}{f'(x_n)} = 0 ,$$ i.e, $f(x^*) = 0.$

\subsection{Robbins and Monro Method} 
The Robbins–Monro (RM) algorithm, proposed in \cite{Robbins51} offer a methodology for finding root of a stochastic function. The function is not defined explicitly and is represented as an expected value of the noisy samples. That is, only noisy observations (measurements) can be made at any desired value. This algorithm is basically a modification of NR algorithm given in Subsection~\ref{Subsec.NR-Method}. Let $M(x)$ be an unknown real-valued, monotonically increasing function with:
\begin{equation*}
    M(x)=\alpha ,
\end{equation*} 
where $\alpha$ is a given constant, in purpose to find the  unique root $x=\upkappa$ of this equation. For determining the value of $\upkappa,$ the first step is the more or less arbitrary choice of one or more values then successively extract new values $x_n$ as certain functions of the previously acquired $x_1,\ldots,x_{n-1},$ the values $M(x_1),\ldots,M(x_{n-1})$ in such a way that:
\begin{equation*}
    \lim_{n\to\infty} x_{n} = \upkappa .
\end{equation*}
The level of convergence in the equation and the ease of computation of x decide the functional effectiveness and the usefulness of the proposed method. Since the function $M(x)$ is unknown to the experimenter, we suppose that to each value x corresponds a random variable (RV) $Y = Y(x)$ with distribution function $Pr[Y(x) \leq y] = H(y|x),$ such that 
\begin{equation*}
    M(x) = \int_{-\infty}^{+\infty} y dH(y|x) ,
\end{equation*} 
is the expected value of $Y$ for the given $x,$ i.e., $ M(x) = \mathbb{E}\{Y_x\}.$
Hence, we can make some successive observations $Y$ at levels $x_1,x_{2},\ldots$ determined sequentially to estimate the desired value of the root $\upkappa.$ We determine then all the needed approximations, starting with the first approximation $X_{1},$ using this recursive formula:
\begin{equation}
    \label{Eq:RM}
    X_{n+1} = X_n-a_n(Y_{X_n}-\alpha),
\end{equation}
where $X_{1}$ a real constant and  $a_{n\{n \in {\mathbb{N}}\}}$ a sequence of real constants satisfying the following conditions: $a_n > 0,\, a_n \to 0,$ and $\sum^{\infty}_{n=0}a_n = \infty$ and $\sum^{\infty}_{n=0}a^2_n < \infty.$
\section{Problem Formulation and Preliminaries}
\label{Sec.III}

In this section, we present the adopted identification test model and establish some preliminaries and definitions required to understand such test.

\subsection{Problem Statement}

In several engineering applications it is typical that an experimenter does not know the exact mathematical model. Instead, S/he observes only a set of random samples from the hypothetical model $M(x).$ We assume that such model is the expected value of such random observations. More accurately, let define the model $M: \mathbb{R} \rightarrow \mathbb{R}$ by $M(x) \triangleq \mathbb{E}(Y_x)$ where $\{Y_x|x \in \mathbb{R}\}$ is the set of RVs. We address an identification-focused test setup, for which the objective of the test is defined as follows: Identifying the roots of a \emph{stochastic} given function, that is, given a finite interval $[a,b],$ determining whether or not a root of the function is in such an interval. To accomplish this task, we assume that realizations of certain RVs are provided, namely, a set of RVs $\{Y_x|x \in \mathbb{R}\}$ is given. Let assume that each observation is denoted by $Y^\upkappa_{x}$ and is equi-distributed on the interval $[x-\upkappa-\delta, x-\upkappa+\delta] ,$ where $\upkappa$ is an unknown parameter and $\delta > 0.$ For a given value of $\upkappa,$ the RV $Y^\upkappa_{x}$ and $Y^\upkappa_{x_{1}}-x_1+x_0$ have the same distribution, thus it is irrelevant which $Y_x$ we consider for the proof analysis. To this end we consider $Y^\upkappa \triangleq Y^\upkappa_{0}.$ Moreover, let $\fY_1^n \triangleq (Y_1,\ldots,Y_n)$ denotes the vector of observation. We assume in the following that $\{Y_i\}_{i=1}^n$ are a class of independent RVs where they have the same distribution as of $Y^\upkappa,$ i.e. $\{Y_i\}_{i=1}^n \overset{\text{\tiny IID}}{\sim} Y^\upkappa.$

\begin{definition}[Test with Two Errors]
A $(n,\varepsilon, \lambda_1, \lambda_2)$-test $\Uptheta$ for an interval $[a,b]$ is a measurable mapping $\Uptheta: \mathbb{R}^n \rightarrow \{0,1\}$ such that
\begin{align}
\label{Eq.Test_I_Type_I}
\Pr \left( \Uptheta \circ \fY_1^n = 0 \right) & \leq \lambda_1 \qquad \forall \upkappa \in [a,b]
\\
\Pr \left( \Uptheta \circ \fY_1^n = 1 \right) & \leq \lambda_2 \qquad \forall \upkappa \in (-\infty,a-\varepsilon] \cup [b+\varepsilon,\infty)
\label{Eq.Test_I_Type_II}
\end{align}
\end{definition}
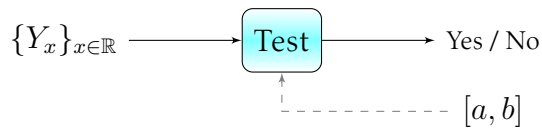
\begin{figure}[H]
    \label{Fig.Test}
    \centering
    \tikzstyle{l} = [draw, -latex']
\tikzstyle{Block1} = [draw, top color = white, middle color = cyan!50, rectangle, rounded corners, minimum height=2em, minimum width=2.5em]
\tikzstyle{input} = [coordinate]
\tikzstyle{sum} = [draw, circle,inner sep=0pt, minimum size=5mm,thick]
\tikzstyle{conv} = [draw, circle,inner sep=0pt, minimum size=5mm,thick]
\tikzstyle{arrow}=[draw,->]
\tikzstyle{small_node} = [draw, circle,inner sep=0pt, minimum size=.8mm,thick]

\begin{tikzpicture}[auto, node distance=2cm,>=latex']

\node[] (M) {$\{Y_{x}\}_{x\in\mathbb{R}}$};
\node[Block1,right=1.5cm of M] (enc) {\text{\fontfamily{jkpl}\selectfont Test}};
\node[right=1.5cm of enc] (Output) {$\text{\fontfamily{jkpl}\selectfont \small Yes\,/\,No}$};
\node[below=3mm of Output] (Target) {$[a,b]$};
\draw[->] (M) -- (enc);
\draw[->] (enc) -- (Output);
\draw[dashed,<-,gray] (enc) |- (Target);

\end{tikzpicture}
    \vspace{-3mm}
    \caption{Root identification setting for a stochastic given function. Given the noisy observations of an unknown function $M(x),$ a test should reliably determine whether or not a root of the function $M(x),$ belongs to the interval $[a,b].$}
\end{figure}
\begin{remark}
In the following we draw some connections between the identification of messages and identification of roots of a function.
\begin{itemize}
    \item Observe that the statistical test $\Uptheta$ might be understood and interpreted as of an identification function with an algorithmic behavior from the set of observations to a binary set of $\{0,1\},$ i.e., the test finally declares either a Yes or No in a reliable manner.
    \item Furthermore, Observe that \eqref{Eq.Test_I_Type_I} and \eqref{Eq.Test_I_Type_II} resemble the conventional notions of the type I\,/\,II error probabilities, respectively, as observed in the context of deterministic and randomized identification via channels \cite{Salariseddigh_PhD_Diss,AD89}. In particular, \eqref{Eq.Test_I_Type_I} describes the stochastic event that the test $\Uptheta$ misses to identify the root of the function, and \eqref{Eq.Test_I_Type_II} represents the event of declaring that the root belongs in $[a,b]$ while the root is not in $[a,b].$
    \item We can interpret the identification of messages as an implicit test, which represent a measurable mapping $\Uptheta : \mathbb{R}^n \to \{0,1\}.$ This test $\Uptheta$ is accomplished for every pair of messages. That is, for every pair, the test verifies whether the output of channels belongs to the target decoder (decoder associated with the selected message in the receiver for identification).
    \end{itemize}
\end{remark}

\begin{definition}[Test with Three Errors]
\label{Def.Test_Three_Errors}
An $(n,\epsilon_1, \epsilon_2, \lambda_1, \lambda_2, \lambda_3)$-test $\Upphi$ for $[a,b]$ is a measurable mapping $\Upphi: \mathbb{R}^n \rightarrow \{0,1\}$ which fulfill the following error conditions
\begin{align}
    \Pr(\Upphi \circ (Y_1,\ldots,Y_n)=1) & > 1- \lambda_1 \, \qquad \forall \theta \in [a,b]
    \\
    \Pr(\Upphi \circ (Y_1,\ldots,Y_n)=0) & > 1- \lambda_2 \, \qquad \forall \theta \not\in (a-\epsilon_2,b+\epsilon_2)
    \\
    \Pr(\Upphi \circ (Y_1,\ldots,Y_n)=1) & > 1- \lambda_3 \, \qquad \forall \theta \in [a-\epsilon_2+\epsilon_1,a-\epsilon_1] \cup [b+\epsilon_1,b+\epsilon_2+\epsilon_1]
\end{align}
where $\{Y_i\}_{i=1}^n$ are independent RVs with the same distribution as $Y^\theta.$
\end{definition}

\begin{definition}[Measure]
    Let $\mathbbmss{X}$ be a set and let $\Sigma$ denotes its $\sigma$-algebra, i.e., nonempty collection of subsets of $\mathbbmss{X}$ closed under complement, countable unions, and countable intersections. A function $m: \Sigma \to \mathbb{R}$ is a measure if the following holds:
    \begin{itemize}
        \item $\forall E \in \Sigma$ the measure's realization is non-negative, i.e., $m (E) \geq 0.$  
        \item Measure of empty set is zero, i.e., $m(\emptyset) = 0.$
        \item For every countable collections $\{E_k\}_{k=1}^{\infty}$ of disjoint sets in $\Sigma,$ $m\left( \bigcup_{k=1}^{\infty} \right) = \sum_{k=1}^{\infty} m(E_k).$
    \end{itemize}
\end{definition}

\begin{theorem}[Intermediate Value Theorem for Continuous Functions; {\cite[Sec.~3.10]{Apostol67}}]
    Consider an interval $I=[a,b]$ in the real numbers $\mathbb{R}$ and a continuous function $ f\colon I\to \mathbb{R} .$ Then if u is a number between f(a) and f(b), that is, $\min(f(a),f(b))<u<\max(f(a),f(b)),$ then there is a $c\in (a,b)$ such that $f(c)=u.$
\end{theorem}

\subsection{Previous Results - Identification Test With Two Errors}
In the following, we first present previous results for identification of roots where only two type of errors are considered. Such test consider only one pair of gap around the target interval. 

\begin{theorem}[see {\cite{KleinewaechterISIT97}}]
        \label{Th.Test_Two_Errors}
        If $\Uptheta$ is an $(n, \varepsilon, \lambda_1, \lambda_2)$-test where $\varepsilon < 2\delta.$ Now, let define $\lambda \triangleq \max\{\lambda_1,\lambda_2\}.$ Then, the number of observations $n$ is lower bounded as follows
        \begin{equation}
        \label{Eq.LBFormula}
            n \geq \frac {\log 2\lambda}{\log (1- \frac{\varepsilon}{2\delta})} .
        \end{equation}
        \end{theorem}

\subsection{Lower Bound Analysis - Identification Test With Two Errors}
We first began by studying the gap between the single and the double logarithmic rate. Further, we tried to examine the shape of the curve of this bound of identification of roots while fixing some parameters and then plotting it in function of the remaining parameter.

\begin{figure}[!htb]
    \centering
    \includegraphics[width=0.7\textwidth]{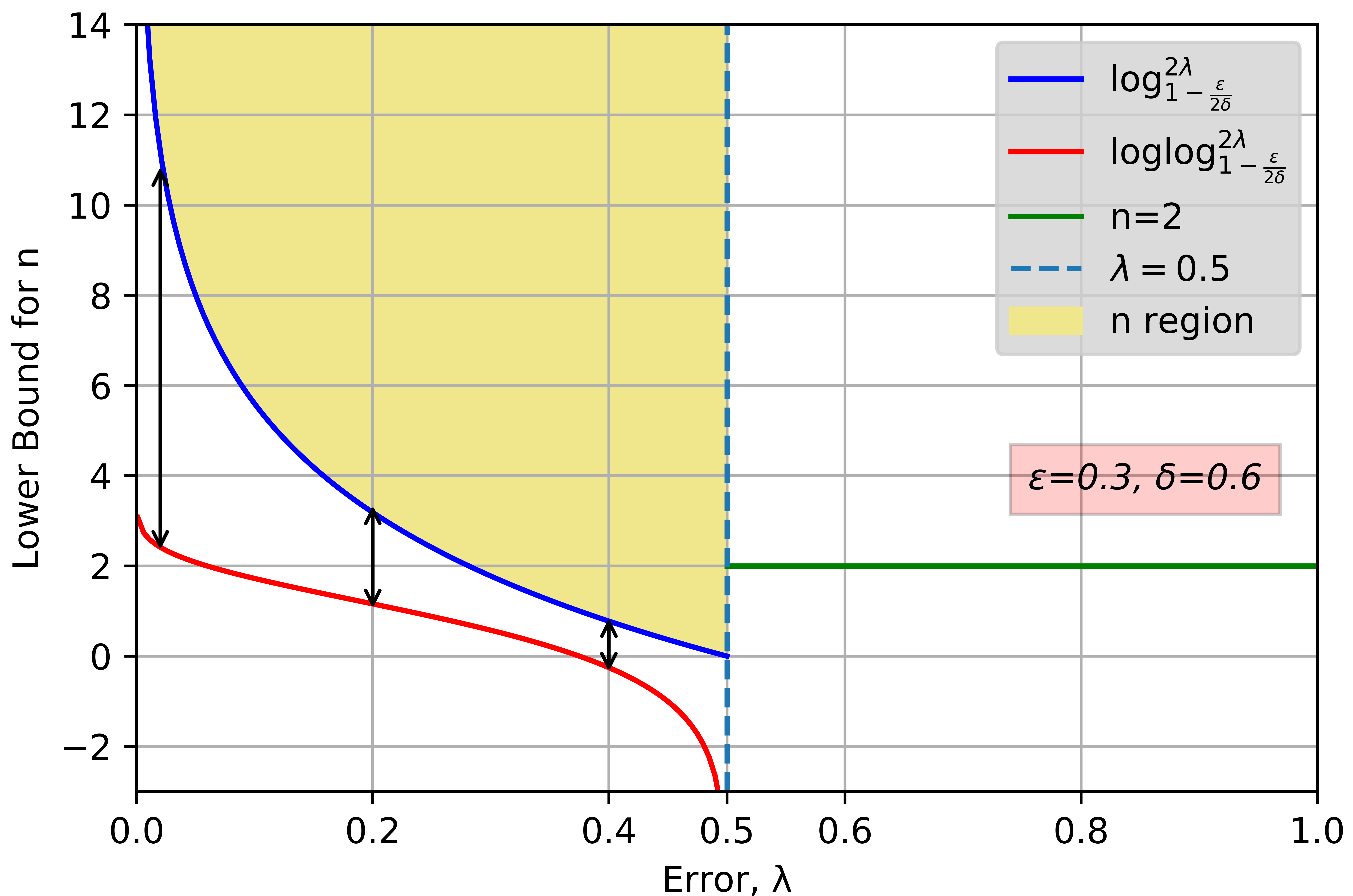}
    \caption{Illustration of first and second order logarithmic lower bounds on the number of observations. The gap between the two curves increases for small values of errors.}
    \label{fig.svsdlog}
\end{figure}

The blue curve in Figure~\ref{fig.svsdlog} corresponds to the lower bound given in \eqref{Eq.LBFormula}. We can observe that for asymptotically small values of $\lambda$ the required number of observation increases arbitrarily large and its gap to the well known double logarithmic rate as observed in the message identification problem \cite{AD89}, marked in red color, monotonically increases. In principle, the gap between the two curves becomes even larger for certain ratio of $\frac{\epsilon}{2\delta}.$ However, in the middle regimes where $\lambda$ is between $0.2$ and $0.4$ we need relatively less number of observations. For example for $\lambda = 0.2$ we need only three observations. On the other hand, for $0.5 < \lambda < 1,$ the root identification test needs evidently only two observations \cite{KleinewaechterISIT97}.

\section{Main Results - Lower Bound on the Number of Observations}
\label{Sec.IV}
In this section we introduce a new test with two pair of gaps in the neighbourhood of a target interval which leads a more general lower bound. Our main theorem is stated as follows.

\begin{theorem}
\label{Th.Test_Three_Errors}
    Let $\Upphi$ be an $(n,\epsilon_1, \epsilon_2, \lambda_1, \lambda_2, \lambda_3)$-test with $0 < 2\epsilon_1 < \epsilon_2 < 2\delta.$ Now, let define $\lambda \triangleq \max\{\lambda_1,\lambda_2,\lambda_3\}.$ Then the following inequality between the test parameters holds
    \begin{equation}
        \label{Ineq.LB_New_Test}
        6 \lambda \geq \Big[ 2(1-\frac{\epsilon_1}{2\delta})^n + (1-\frac{\epsilon_2}{2\delta})^n \Big] .
    \end{equation}
\end{theorem}

\begin{proof}
    The proof of lower bound given in Theorem~\ref{Th.Test_Three_Errors} is provided as follows. Observe that for all $c \in [a,b]$ we have
    \begin{equation}
        \label{Eq.NT1}
         m (\Upphi^{-1}(1) \cap [c-\delta,c+\delta]^n) > (1- \lambda_1) \cdot (2\delta)^n ,
    \end{equation}
and for all $c \notin (a-\epsilon_2,b+\epsilon_2)$ we have
\begin{equation}
    \label{Eq.NT2}
    m (\Upphi^{-1}(0) \cap [c-\delta,c+\delta]^n) > (1- \lambda_2)\cdot(2\delta)^n .
\end{equation}
Finally, for all $c \in [a-\epsilon_2+\epsilon_1,a-\epsilon_1] \cup [b+\epsilon_1,b+\epsilon_2+\epsilon_1]$ it holds
\begin{equation}
    \label{Eq.NT3}
    m (\Upphi^{-1}(1) \cap [c-\delta,c+\delta]^n) > (1- \lambda_3)\cdot(2\delta)^n .
\end{equation}
Observe that setting $\theta = c$ and exploiting \eqref{Eq.NT1} we obtain
    \begin{align*}
   \Pr(\Upphi \circ (Y_1,\ldots,Y_n)=1) > 1- \lambda_1 .
    \end{align*}
    Observe that since $Y_i$ are independent equi-distributed on $[c-\delta,c+\delta],$ therefore $(Y_1, \dots, Y_n)$ is equi-distributed on $[c-\delta,c+\delta]^n,$ i.e., we have
    \begin{align*}
    \Pr(\Upphi \circ (Y_1,\ldots,Y_n)=1 \land (Y_1,\ldots,Y_n) \in [c-\delta,c+\delta]^n) > 1- \lambda_1 ,
    \end{align*}
    which implies
    \begin{align*}
       \Pr(\Upphi \circ (Y_1,\ldots,Y_n)=1 \land (Y_1,\ldots,Y_n) \in [c-\delta,c+\delta]^n) \Rightarrow \frac{m (\Upphi^{-1}(1) \cap [c-\delta,c+\delta]^n)}{m([c-\delta,c+\delta]^n)} .
    \end{align*}
    Therefore,
    \begin{align*}
        m (\Upphi^{-1}(1) \cap [c-\delta,c+\delta]^n) > (1- \lambda_1)\cdot m([c-\delta,c+\delta]^n) = (1- \lambda_1)\cdot(2\delta)^n .
    \end{align*}
    The two inequalities \ref{Eq.NT2} and \ref{Eq.NT3} can be proved similarly.
Next, we distinguish between six cases. Let start with the first two cases $\theta = a$ and $\theta = a-\epsilon_1.$ Let $\Upphi$ be an $(n,\epsilon_1, \epsilon_2, \lambda_1, \lambda_2, \lambda_3)$-test for $\theta \in [a,b].$ Let
\begin{align*}
    J_1 = [a-\delta,a+\delta]^n \quad \text{ and } \quad N_1 = [a-\epsilon_1-\delta,a-\epsilon_1+\delta]^n .
\end{align*}
Exploiting \eqref{Eq.NT1} we obtain
\begin{align*}
    m (\Upphi^{-1}(1) \cap J_1) > (1- \lambda_1)\cdot(2\delta)^n \quad \text{ and } \quad m (\Upphi^{-1}(0) \cap N_1) > (1- \lambda_1)\cdot(2\delta)^n .
\end{align*}
On the other hand we have
\begin{align*}
    m (\Upphi^{-1}(1) \cap J_1\cap N_1) \leq m (J_1 \cap N_1) \leq (2\delta - \epsilon_1)^n .   
\end{align*}
Therefore, since $\lambda_1 \leq \lambda$ we obtain
\begin{align*}
    (1- \lambda)\cdot(2\delta)^n &< m (\Upphi^{-1}(1) \cap J_1) \\
    &= m (\Upphi^{-1}(1) \cap J_1\cap N_1) + m (\Upphi^{-1}(1) \cap (J_1 \setminus N_1)) \\
    &< m (\Upphi^{-1}(1) \cap J_1\cap N_1) + m (J_1\setminus N_1) \\
    &= m (\Upphi^{-1}(1) \cap J_1\cap N_1) + (2\delta)^n - (2\delta - \epsilon_1)^n .
\end{align*}
Thereby,
\begin{align*}
    m (\Upphi^{-1}(1) \cap J_1\cap N_1) & \geq (1- \lambda) (2\delta)^n - (2\delta)^n + (2\delta - \epsilon_1)^n = (2\delta - \epsilon_1)^n - \lambda (2\delta)^n .
\end{align*}
Similarly, we have
\begin{align*}
    m (\Upphi^{-1}(0) \cap J_1\cap N_1) & \geq (1- \lambda) (2\delta)^n - (2\delta)^n + (2\delta - \epsilon_1)^n \\
    & = (2\delta - \epsilon_1)^n - \lambda (2\delta)^n .
\end{align*}
Now we note that the test also has to distinguish between $a-\epsilon_2+\epsilon_1$ and $a-\epsilon_2.$ As the two intervals have the same size $\epsilon_1$ we can similarly obtain
\begin{align*}
    J_2 = [a-\epsilon_2+\epsilon_1-\delta,a-\epsilon_2+\epsilon_1+\delta]^n \quad \text{ and } \quad N_2 = [a-\epsilon_2-\delta,a-\epsilon_2+\delta]^n .
\end{align*}
Therefore, since $\lambda_2 \leq \lambda$ we obtain
\begin{align*}
    m (\Upphi^{-1}(1) \cap J_2\cap N_2) & \geq (2\delta - \epsilon_1)^n - \lambda (2\delta)^n .
\end{align*}
Similarly, we have
\begin{align*}
    m (\Upphi^{-1}(0) \cap J_2\cap N_2) &\geq (2\delta - \epsilon_1)^n - \lambda (2\delta)^n .
\end{align*}
For the last two cases the test $\Upphi$ has to distinguish between $a$ and $a-\epsilon_2.$ We get similarly 
\begin{align*}
    J_3 = [a-\delta,a+\delta]^n \quad \text{ and } \quad N_3 = [a-\epsilon_2-\delta,a-\epsilon_2+\delta]^n .
\end{align*}
Therefore, since $\lambda_3 \leq \lambda$ we obtain
\begin{align*}
    \Rightarrow m (\Upphi^{-1}(1) \cap J_3\cap N_3) &\geq (2\delta - \epsilon_1)^n - \lambda (2\delta)^n .
\end{align*}
Next, similarly we obtain
\begin{align*}
    m (\Upphi^{-1}(0) \cap J_3\cap N_3) &\geq (2\delta - \epsilon_1)^n - \lambda (2\delta)^n .
\end{align*}
Thus 
\begin{align*}
    m (\Upphi^{-1}(1) \cap J_1\cap N_1) + m (\Upphi^{-1}(0) \cap J_1\cap N_1)+
    m (\Upphi^{-1}(1) \cap J_2\cap N_2)\\ + m (\Upphi^{-1}(0) \cap J_2\cap N_2) + m (\Upphi^{-1}(1) \cap J_3\cap N_3) + m (\Upphi^{-1}(0) \cap J_3\cap N_3)\\\leq m (J_1\cap N_1) + m (J_2\cap N_2) + m (J_3\cap N_3) = 2(2\delta - \epsilon_1)^n + (2\delta - \epsilon_2)^n .
\end{align*}
Therefore,
\begin{align*}
    2(2\delta - \epsilon_1)^n +  (2\delta - \epsilon_1)^n & \geq 4((2\delta - \epsilon_1)^n - \lambda (2\delta)^n) + 2((2\delta - \epsilon_2)^n - \lambda (2\delta)^n)\\&
    \Leftrightarrow  2(2\delta - \epsilon_1)^n +  (2\delta - \epsilon_1)^n \geq 4(2\delta - \epsilon_1)^n - 6\lambda (2\delta)^n + 2(2\delta - \epsilon_2)^n\\
    &\Leftrightarrow 6\lambda \geq \frac{2(2\delta - \epsilon_1)^n + (2\delta - \epsilon_2)^n}{(2\delta)^n}\\
    & \Leftrightarrow 6\lambda \geq 2\big(1-\frac{\epsilon_1}{2\delta}\big)^n + \big(1-\frac{\epsilon_2}{2\delta}\big)^n \\
    & \Leftrightarrow \log (6\lambda) \geq \log \Big( 2 \big( 1 - \frac{\epsilon_1}{2\delta} \big)^n + \big( 1-\frac{\epsilon_2}{2\delta} \big)^n \Big) .
    \end{align*}
\end{proof}

\begin{corollary}
    For $\epsilon_1 = \epsilon_2$ and $\lambda_1 = \lambda_2 = \lambda_3,$ we recover the previous identification test, i.e., 
\begin{align}
    6\lambda \geq 2 \big(1-\frac{\epsilon_1}{2\delta} \big)^n + \big(1-\frac{\epsilon_2}{2\delta} \big)^n &\xRightarrow{\text{$\epsilon_1 = \epsilon_2$}} 6\lambda \geq 3\big( 1-\frac{\epsilon_1}{2\delta} \big)^n
    \nonumber\\
    & \Leftrightarrow \log (6\lambda) \geq \log \Big(3 \big(1-\frac{\epsilon_1}{2\delta}\big)^n\Big)
    \nonumber\\
    & \Leftrightarrow \frac{\log 2\lambda}{\log \big(1 - \frac{\epsilon}{2\delta} \big)} \leq n .
\end{align}
\end{corollary}
\begin{corollary}
    Observe that similar to the arguments given in the old test the condition here for $n = \infty$ is given by $\epsilon_1, \epsilon_2 \to 0.$ To show this we simplify the lower bound when $\epsilon_1 \to 0,$ i.e,
\begin{align}
    6\lambda \geq 2\big(1-\frac{\epsilon_1}{2\delta}\big)^n + \big(1-\frac{\epsilon_2}{2\delta}\big)^n &\xRightarrow{\text{$\epsilon_1 = 0$}} 6\lambda \geq 2 + \big(1-\frac{\epsilon_2}{2\delta}\big)^n .
\end{align}
Thereby, $\log (6\lambda-2) \geq n\log (1-\frac{\epsilon_2}{2\delta})^n,$ which implies
\begin{align}
    n \geq \frac{\log (6\lambda-2)}{\log (1- \frac{\epsilon_2}{2\delta})} .
\end{align}
Now setting $\epsilon_2 \to 0,$ the number of observations, $n,$ tends to $\infty.$ Furthermore, in order to compare the new test versus the previous one in \eqref{Eq.LBFormula} for $\epsilon_2=\epsilon$ we obtain
 \begin{align}
    \frac{\log(6\lambda-2)}{\log(1-\frac{\epsilon_2}{2\delta})}\leq \frac{\log(2\lambda)}{\log(1-\frac{\epsilon}{2\delta})} \rightarrow 6\lambda -2 \leq 2\lambda \rightarrow \lambda \geq \frac{1}{2} .
\end{align}
This implies that for the range of $\lambda \geq \frac{1}{2}$ the test with 3 errors gives an improved lower bound in terms of the required number of observation. However, in order that the inequality $\log(6\lambda-2)/  \log(1-\frac{\epsilon_2}{2\delta}) > 0$ holds, the term $\log(6\lambda-2)$ should be negative because $\log (1-\frac{\epsilon}{2\delta}) < 0.$
\begin{align*}
    \log (6\lambda-2) < 0 \Rightarrow 0 < 6\lambda-2 < 1 \Leftrightarrow \lambda < \frac{1}{2} .
\end{align*}
Moreover, another condition on $\lambda$ is imposed, namely, if $6\lambda-2 \leq 0$ then $\log (6\lambda-2) \to \infty,$ so this logarithm is not defined, that is why $6\lambda-2 > 0$ should hold  which implies $\lambda > \frac{1}{3}.$ Therefore, if $\epsilon_1 = 0$ the new test does not provide advantage.
\end{corollary}
\begin{remark}[Complexity of the root identification]
As Figure~\ref{fig:Motiv} shows in contrast to identification of messages where a double logarithmic rate is possible, our finding in this works confirms that the problem of identification of roots of a function does not admit such fast rate. In fact, the root identification is a logarithmic problem and much slower compared to the message identification problem. However, the root identification is faster than finding the root of a function. For example, the complexity of finding the root for example in the case Kiefer-Wolfowitz algorithm \cite{Kiefer52} is $\mathcal{O}(n^{-1})$ but the complexity for identifying the root is logarithmic.

\begin{figure}[H]
    \centering
    \begin{tikzpicture}[x=1cm,y=0.5cm]
        \draw[-stealth] (-2,0)--(10,0) node [right]{n}; 
        \draw[dashed] (-1.5,-2.5)--(-1.5,2.5) node (a) [right]{$\sim \log\log N(n,\lambda)$ };
        \node[below = 1.6cm of a]{\hspace{1cm} \small message identification};
        \draw[dashed] (3.5,-2.5)--(3.5,2.5) node (b) [right]{$\log_{1-\frac{\epsilon}{2\delta}}^{2\lambda}$ };
        \node[below =1.5cm of b]{\hspace{1cm} \small root identification};
        \draw[dashed] (6.5,-2.5)--(6.5,2.5) node (c) [right]{$\mathcal{O}(n^c)$ };
        \node[below =1.6cm of c]{\hspace{1.5cm} \small root computation};
        \draw [very thick, mycolor9] (3.5,0)--(4,1) ;
        \draw [very thick, mycolor9] (3.5,-1)--(4.5,1) ;
        \draw [very thick, mycolor9] (4,-1)--(5,1) ;
        \draw [very thick, mycolor9] (4.5,-1)--(5.5,1) ;
        \draw [very thick, mycolor9] (5,-1)--(6,1) ;
        \draw [very thick, mycolor9] (5.5,-1)--(6.5,1) ;
        \draw [very thick, mycolor9] (6,-1)--(6.5,0) ;

        \draw [thick, blau_2b] (9,-1)--(8,1) ;
        \draw [thick, blau_2b] (8.5,-1)--(7.5,1) ;
        \draw [thick, blau_2b] (8,-1)--(7,1) ;
        \draw [thick, blau_2b] (7.5,-1)--(6.5,1) ;
        \draw [thick, blau_2b] (7,-1)--(6.5,0) ;
\end{tikzpicture}
    \caption{Spectrum of algorithmic bounds for message identification \cite{AD89}, root identification and root computation \cite{Ahlswede87}.}
    \label{fig:Motiv}
\end{figure}
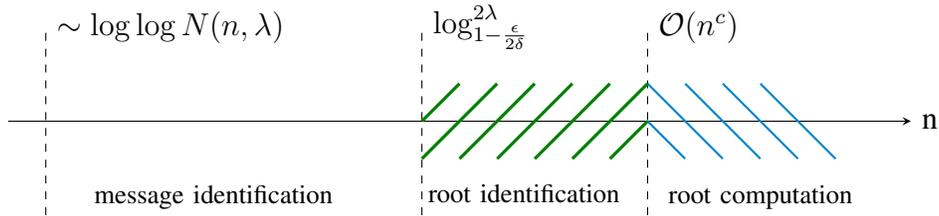
\end{remark}

\subsection{Analysis of The New Lower Bound}
In the following, we analyse the new test and provide several sketch of the inequality provided in Theorem~\ref{Th.Test_Three_Errors}.

\begin{figure}[t]
    \centering
    \setkeys{Gin}{width=0.5\textwidth}
    \subfloat[Lower bounds of previous and new test for $\epsilon_1 < \epsilon$\label{Fig.NewvsOldTest}]{\includegraphics[width=0.5\textwidth]{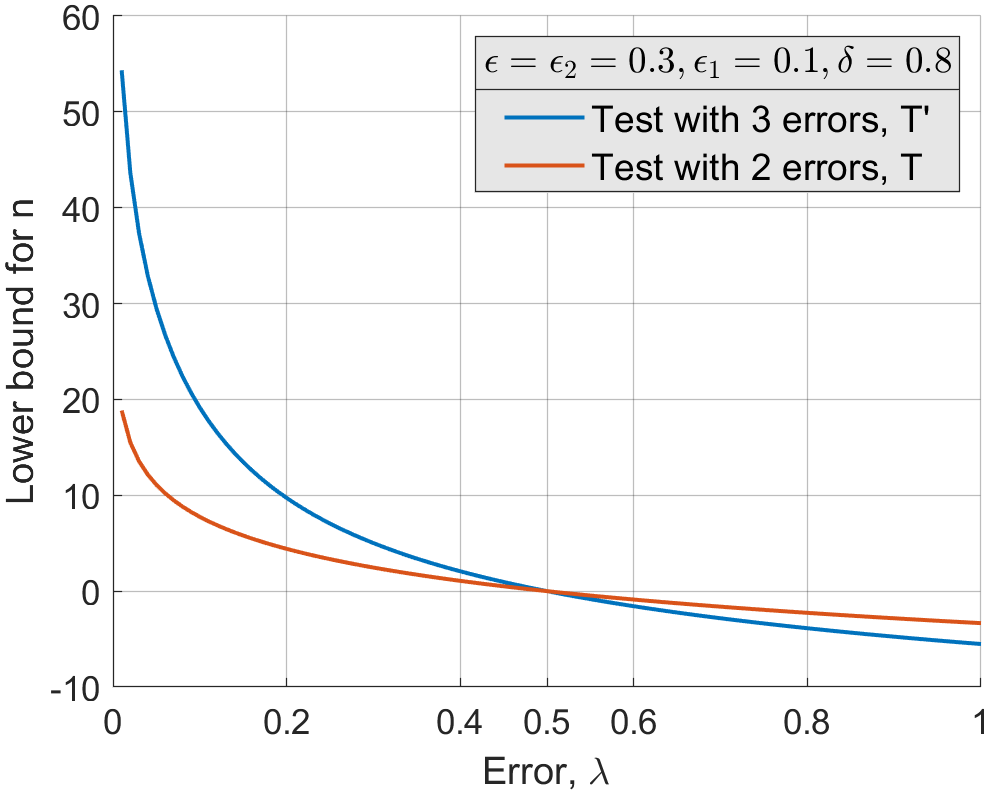}}
    \hfill
    \subfloat[Lower bounds of previous and new test for $\epsilon_1 \geq \epsilon$ with $\epsilon_1 \geq \epsilon$\label{Fig.NewvsOldTest_3}]{\includegraphics[width=0.5\textwidth]{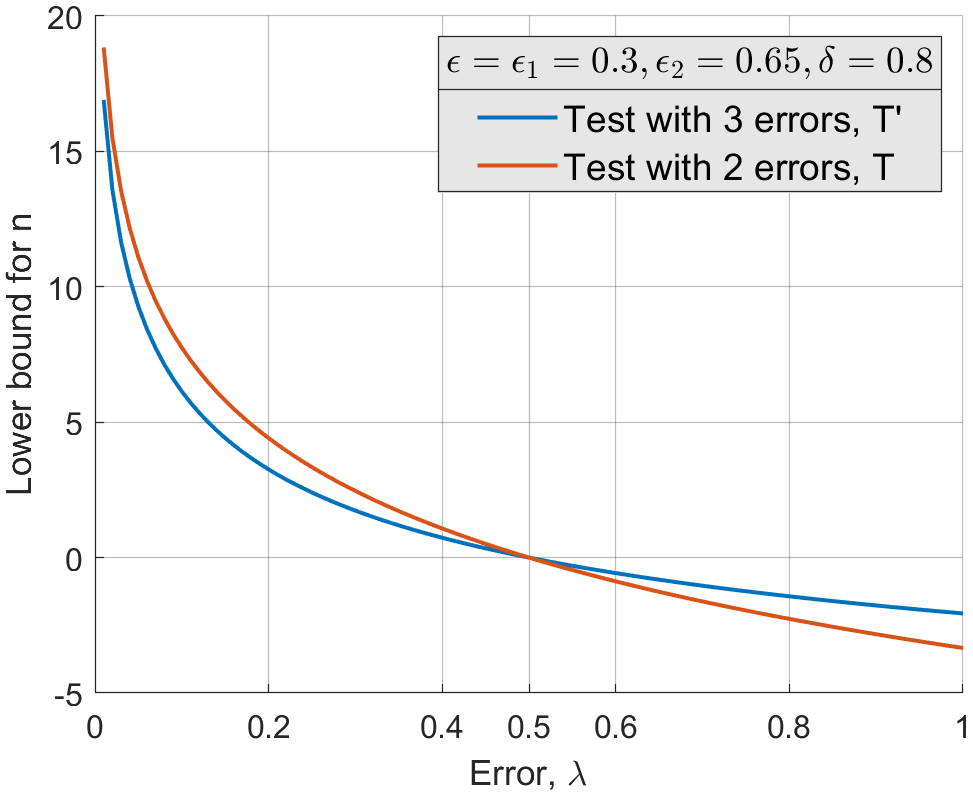}}
    \hfill
    \subfloat[Lower bounds of previous and new test for $\epsilon_1 \geq \epsilon$ with decreasing parameters \label{Fig.NewvsOldTest_3_gap}]{\includegraphics[width=0.5\textwidth]{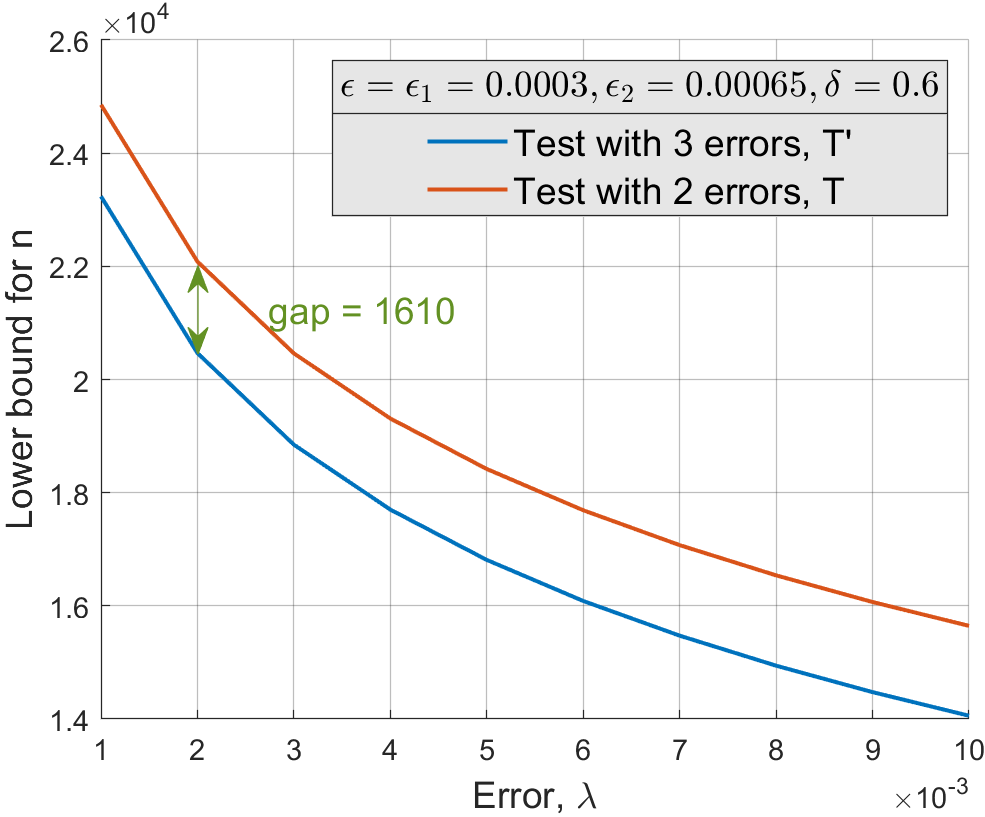}}
    \hfill
    \subfloat[Lower bounds of previous and new test where $\epsilon_1 = 0$\label{Fig.NewvsOldTest_eps1_zero}]{\includegraphics[width=0.5\textwidth]{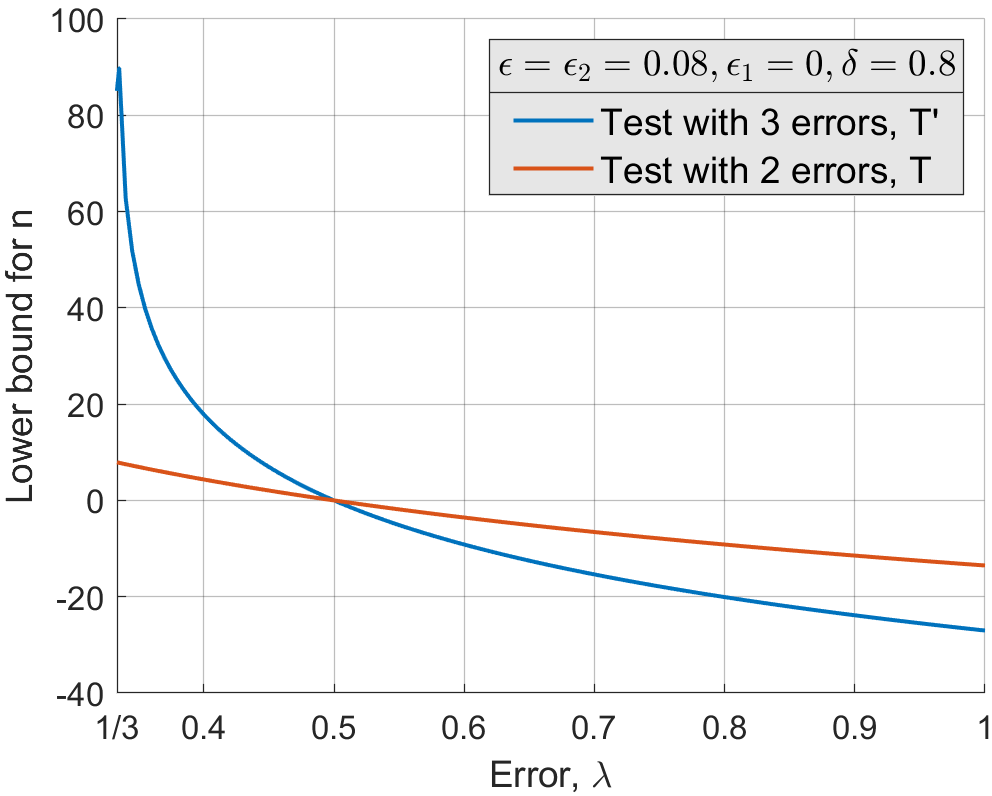}}
    \caption{Comparison of previous and new test for root identification across different test parameters.}
\end{figure}

\begin{itemize}
    \item $\epsilon_1 < \epsilon:$ As Figure \ref{Fig.NewvsOldTest} shows the blue curve is higher than the red one. So the new lower bound found is still higher than the one found with the previous Test. So win this case it is not a better lower bound.
    \item$\epsilon_1 \geq \epsilon:$ As Figure \ref{Fig.NewvsOldTest} shows the blue curve is this time lower than the red one. So we recognize here a better lower bound. This gap does in fact increases when the values of $\epsilon_1, \epsilon_2$ and $\epsilon_3$ decreases. As confirmed by Figure \ref{Fig.NewvsOldTest_3}, for low number of error, $\lambda,$ the previous test requires way more observations, $1610.$ This proves that in a small range of the parameters the new test gives improved result
    \item $\epsilon_1 \leq \epsilon:$ and if the root is not in the silent area in the previous test then the new test outperforms the previous one and requires smaller number of observations.
    \item When the root is in the silent gap of the previous test then there are two possibilities:
    \begin{itemize}
        \item if $\epsilon_1 \ll \epsilon$ then this root will turns out to a non silent region for new Test and we are able to identify it. However the number of $n$ increases.
        \item if $\epsilon_1 \leq \epsilon$ such that the root still in the silent area, which require less number of observations but more complexity and we still cannot identify the root. Hence, as long as the root remains in the non silent area and $\epsilon_1 \leq \epsilon $ the new Test is better.
    \end{itemize}
\item $\lambda > 0.5:$ As Figure~\ref{Fig.NewvsOldTest} shows the red curve which represent the old test is above the blue one which speaks for the new test while in Figure~\ref{Fig.NewvsOldTest_3} the blue curve is above the red one. However in both case they are in the negative part of $n$ which doesn't make much sense since $n$ has to be always positive. We have seen in fact that in the old test $\Uptheta$ in this irregular case that only two observations are needed. So we claim and we hope that by the same arguments applied, that the still some constant is needed.
\end{itemize}

\section{Identification of Roots vs. Identification of Messages}
\label{Sec.V}
In this section, we establish a comparison between identification of messages and identification or roots in terms of different performance parameters. Figures~\ref{Fig.OldTest} and ~\ref{Fig.NewTest} show an interpretation of root identification problem in terms of decoding problem.

\begin{figure}[H]
    \centering
    \includegraphics[width=0.7\textwidth]{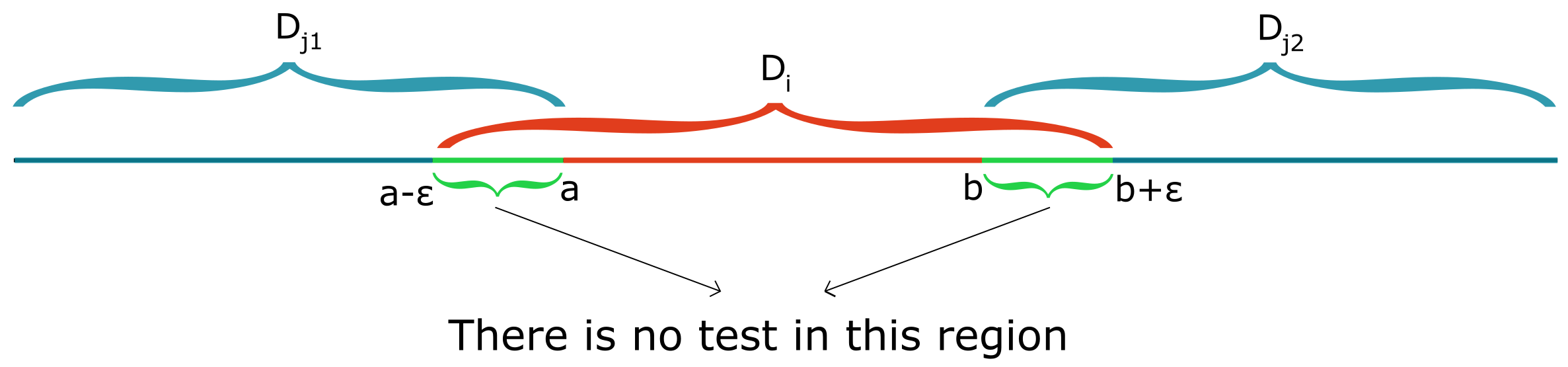}
    \caption{Illustration of identification of roots as decoding problem for the previous test.}
    \label{Fig.OldTest}
\end{figure}

We interpret the silent gap between two intervals as similar to the intersection region between decoding regions in the identification problem \cite{AD89}. Here, $\varepsilon$ is the size of the area around given interval $[a,b].$ This can be justified in the following sense: from one side without this silent area the test $\Uptheta$ would have no meaning (!) and from the other side, this area allow the algorithm to make errors.

\begin{figure}[H]
    \centering
    \includegraphics[width=0.8\textwidth]{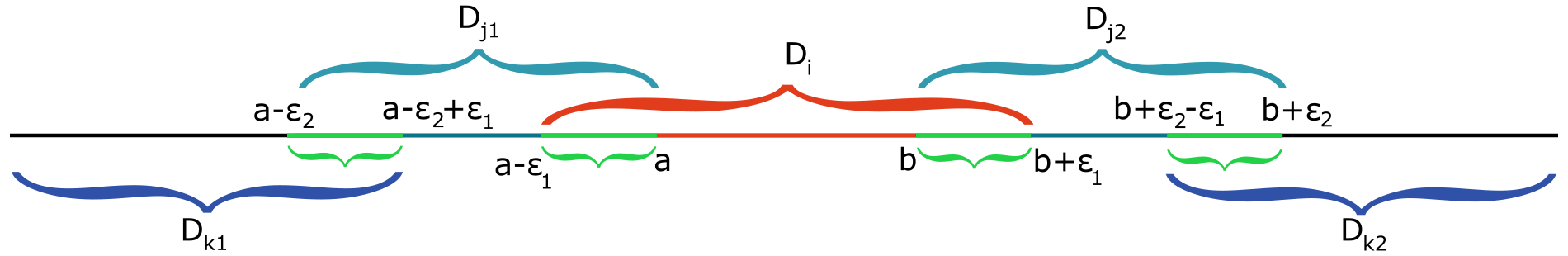}
\caption{Illustration of partitions for new test into different decoding regions.}
\label{Fig.NewTest}
\end{figure}

\begin{table}
\label{Table.Comparison}
    \caption{Comparison of performance parameters}
    \begin{minipage}{.5\linewidth}
      \caption{Identification of messages vs. previous test}
      \centering
      \begin{tabular}{|c|c|}
        \hline
        Identification of Messages & Identification of Roots
        \\
        \hline
        Codeword & Noisy Observation (RV)  \\
        \cellcolor{cell} Codeword Length & \cellcolor{cell} \# of Observations  \\
        Target Message & Target Root\\
        \cellcolor{cell}Errors: $\lambda_1,$ $\lambda_2$ & \cellcolor{cell} Errors: $\lambda_1,$ $\lambda_2$ \\ 
        $\D_i$ & $(a-\epsilon,b+\epsilon)$ \\
       \cellcolor{cell}$\D_{j} = \D_{j1} \cup \D_{j2}$ &  \cellcolor{cell}$(-\infty,a-\epsilon] \cup [b+\epsilon, +\infty) $ \\
        $\D_i \cap \D_j$ &  2 Silent Gaps \\
        \cellcolor{cell} $Q_i(x^n) W^n(y^n|x^n)$  & \cellcolor{cell} Distribution of $Y^\theta_x$ \\ 
        \hline
    \end{tabular}
    \end{minipage}
    \begin{minipage}{.5\linewidth}
      \centering
        \caption{Identification of messages vs. new test}
        \begin{tabular}{|c|c|}
        \hline
        Identification of Messages & Identification of Roots \\
        \hline
        \cellcolor{cell}Errors: $\lambda_1,$ $\lambda_2$ & \cellcolor{cell}Errors: $\lambda_1,$ $\lambda_2,$ $\lambda_3$ \\ 
        $\D_i$ & $(a-\epsilon_1,b+\epsilon_1)$ \\
        \cellcolor{cell}$\D_{j}$ &  \cellcolor{cell}$(a-\epsilon_2,a) \cup (b, b+\epsilon) $ \\
        $\D_{k}$ &  $(-\infty,a-\epsilon_2] \cup [b+\epsilon_2, +\infty) $ \\
        \cellcolor{cell} $\D_i \cap \D_j$ & \cellcolor{cell} 2 Inner Silent Gaps \\
        $\D_j \cap \D_k$ &  2 Outer Silent Gaps \\
        \hline
    \end{tabular}
    \end{minipage}
\end{table}

Next, we introduce the conventional approach for code construction in the identification of messages \cite{Ahlswede89} and then discuss about its differences to the root identification.

\subsection{Coloring Scheme - A Method to Construct the Identification Codes}
In the following, we review a classic method used for the code construction which is proposed in \cite{Ahlswede89}. This review is accomplished for the purpose of understanding an intrinsic property exploited for the code construction in the identification of messages called coloring property.

In \cite{Ahlswede89}, an identification code is constructed via concatenation of two Shannon transmission codes $\C'$ and $\C''.$ By Shannon's coding theorem \cite{S48}, we know that for every $0< \epsilon < C,$ there exists $\delta > 0$ such that for sufficiently large $n$ there is an $(n, M',2^{-n\delta})$-transmission code $\C' = \{(u'_j,\D'_j\,|\,j = 1,\ldots,M')\}$ and an $(\ceil{\sqrt{n}}, M'',2^{-\sqrt{n}\delta})$-transmission code $\C'' = \{(u''_k,\D''_k\,|\,k = 1,\ldots,M'')\}$ where codebook sizes read $M' = \ceil{2^{n(C-\epsilon)}}$ and $M'' = \ceil{2^{\epsilon\sqrt{n}}}.$ Now consider a set of mappings $\{T_i\,|\,i=1,\ldots,N\}, \, T_i : \{1,\ldots,M'\} \rightarrow \{1,\ldots,M''\}$ for all $i \in \{1,\ldots,N\}.$ The mapping $T_i(j)$ which is known to both sender and receiver, is referred to as the color of codeword $i$ under coloring number $j.$ Thus we could construct an identification code $\{(Q(.|i),\D_i)\,|\,i = 1,\ldots,N\}.$ Here $Q(.|i)$ is the uniform distribution on the set of codewords $\U_i = \{u'_j.u''_{T_i(j)}\,|\,j = 1,\ldots, M'\} \subset \X^m$ and a collection of decoding regions $\D_i = \bigcup\limits_{j=1}^{M'} \D'_j \times \D''_{T_i(j)}.$
We select randomly an identification code with such a structure in the following manner:
\begin{equation*}
    Q_i(x^n) = 
    \begin{cases}
    \frac{1}{M'} & \text{if } \exists j : x^n = u'_j.u''_{T_i(j)} \\
    0 & \text{otherwise}
    \end{cases}
\end{equation*}
This means that a coloring can be selected uniformly at random and then the corresponding color of the messages is calculated. The coloring method offers an heuristic rationale for why it is not possible to attain a double logarithmic rate when identifying roots: to decrease the likelihood of the worst-case scenario, we use a set of mostly independent hash functions, i.e., $N$ function $T_i$ from which we select one and send it along with its hash value. That is, we try to break down every potential structure and flatten all possible input characteristics. Thus, using this method to identify roots would be detrimental to the algorithm since it alters the input's original structure.
\section{Conclusions and Discussions}
\label{Sec.VI}

\textbf{Discussions:}
A widely known problem in computational complexity is the exponential gap in time consumption between identifying/verifying a solution and actually finding that solution, particularly evident in the case of \textit{prime factorization} \cite{Crandall05}. For example, it is feasible to verify in polynomial time whether a specified prime number is a factor of the product of two primes. However, to date, no algorithm has been discovered that can perform the factorization itself in polynomial time. When considering the number of inputs as a metric for the size of the problem, it can be asserted that verifying the factorization is a logarithmic problem, whereas the process of finding the factorization is classified as a polynomial problem. This example is similar to our case here, where we found out that identifying the result is not much faster than calculating it and such exponential gap in the speed is not feasible.

\textbf{Conclusions:}
In this study, we enhanced the previous results for logarithmic lower bound on the number of observations by proposing a new test for root identification. We revisit the message identification problem of Ahlswede \cite{AD89} and devise adopted test for identifying roots of a stochastic function. We studied the root identification problem for a special class of functions where only the stochastic knowledge of function is available. A lower bound for number of observation derived in \cite{KleinewaechterISIT97} showed this quantity scales only logarithmic in the test accuracy and error bound. As our key finding, we established a new logarithmic lower bound in the parameters which recovers the previous result as its special case and outperform in several cases. In contrast to identification of messages \cite{AD89}, this logarithmic growth reveal an exponential gap. We examined the possible reason for existence of such gap with respect to the number of observations, namely, a method used in the achievability of identification problem called \textit{coloring scheme} can not be utilized here. We observed that this difficulty arise from considering the uniform distribution as assumed for the identification problem, here flatten the structure of input of the algorithm. That is, using the conventional techniques developed for the achievability construction, e.g., coloring and hash function methods are ineffective in the context of root identification since the original structure of the input is not preserved by simplifying/flattening the algorithm's structure. 

\section*{}
\bibliographystyle{IEEEtran}
\bibliography{Lit}

\end{document}